\documentclass[copyright,creativecommons]{eptcs}
\providecommand{\event}{Gabriel Ciobanu (Ed.): Membrane Computing and\\ 
Biologically Inspired Process Calculi 2009} 
\providecommand{\volume}{15}   
\providecommand{\anno}{2009}   
\providecommand{\firstpage}{1} 
\providecommand{\eid}{1}       
\usepackage{breakurl}

\usepackage{color}

\newcommand{\myway}[1]{\raisebox{-4pt}{\rule{0pt}{16pt}}\colorbox[rgb]{.7,.7,.7}{#1}}
\usepackage{graphicx}
\usepackage{amsmath}
\usepackage{amssymb}
\usepackage{amsthm}
\usepackage{enumerate}

\newtheorem{theorem}{Theorem}

\newtheorem{corollary}[theorem]{Corollary}

\renewenvironment{proof}{\vspace{-3.5ex} \trivlist \item[\hskip \labelsep{\bf Proof.}]}%
{\qed}

\newenvironment{definition}[1][]
   {\refstepcounter{theorem} \par\medskip\noindent
   {\bf Definition~\thetheorem~(#1)} \ignorespaces }
   {\par\medskip }

\newcounter{algcnt}
\newenvironment{algorithm}[1]
{\refstepcounter{algcnt} 
\newcommand{\Precondition}[1]{~\newline 
\item[\hskip\labelsep{\textbf{Precondition:}}] ##1}%
\newcommand{\Postcondition}[1]{
\item[\hskip\labelsep{\textbf{Postcondition:}}] ##1}
\newcommand{\Rules}{~\newline\noindent\textbf{Rules}\nopagebreak}
\begin{trivlist}
\setlength{\itemsep}{-\parsep}
\item[\hskip \labelsep{\framebox{\textbf{Algorithm~\thealgcnt~(#1)}}}]
}%
{\end{trivlist}
}
\renewenvironment{proof}
{~\newline \noindent \emph{Proof.}}
{\vspace*{-2ex}\qed~\newline}

\theoremstyle{remark}
\newtheorem{remark}[theorem]{Remark}

\newtheorem{example}[theorem]{Example}

\begin{document}

\title{New Solutions to the Firing Squad Synchronization Problem for Neural and Hyperdag P~Systems}
\author{Michael~J.~Dinneen \and Yun-Bum~Kim \and Radu~Nicolescu
\institute{Department of Computer Science, University of Auckland,\\
Private Bag 92019, Auckland, New Zealand
\email{\{mjd,radu\}@cs.auckland.ac.nz \ tkim021@aucklanduni.ac.nz}}}
\def\titlerunning{Firing Squad Synchronization Problem for P~Systems}
\def\authorrunning{M.J.~Dinneen, Y.-B.~Kim \& R.~Nicolescu}
\maketitle


\begin{abstract}
We propose two uniform solutions to an open question: the Firing 
Squad Synchronization Problem (FSSP), for hyperdag and symmetric
neural P~systems, with anonymous cells. 
Our solutions take $e_c+5$ and $6e_c+7$ steps, respectively, 
where $e_c$ is the eccentricity of the commander cell
of the dag or digraph underlying these P~systems.
The first and fast solution is based on a novel proposal, 
which dynamically extends P~systems with mobile channels.
The second solution is substantially longer, but is solely based 
on classical rules and static channels.
In contrast to the previous solutions, which work for tree-based P~systems, 
our solutions synchronize to any subset of the underlying digraph;
and do not require membrane polarizations or conditional rules, 
but require states, as typically used in hyperdag and neural P~systems. 
\end{abstract}



\section{Introduction}

The \emph{Firing Squad Synchronization Problem} (FSSP) 
\cite{Kobayashi2005,Mazoyer1987,Noguchi2004,Schmid2004,Szwerinski1982,Umeo2002} 
is one of the best studied problems for cellular automata.
The problem involves finding a cellular automaton, 
such that, after a command is given, all the cells, after some finite time,
enter a designated \emph{firing} state \emph{simultaneously} and \emph{for the first time}.
Several variants of FSSP \cite{Schmid2004,Szwerinski1982}, have been proposed and studied. 
Studies of these variations mainly focus on finding a solution with as few 
states as possible and possibly running in optimum time.

There are several applications that require synchronization.  
We list just three here.
At the biological level, cell synchronization is a process by which cells 
at different stages of the cell cycle (division, duplication, replication) 
in a culture are brought to the same phase.
There are several biological methods used to synchronize cells at specific 
cell phases \cite{Humphrey2005}.  Once synchronized, monitoring the 
progression from one phase to another allows us to calculate the timing 
of specific cells' phases.
A second example relates to operating systems \cite{Silberschatz2004}, 
where process synchronization is the coordination of simultaneous threads or 
processes to complete a task without race conditions.  
Finally, in telecommunication networks \cite{Freeman2005}, 
we often want to synchronize computers to the same time, 
i.e., primary reference clocks should be used to avoid clock offsets. 

The synchronization problem has recently been studied in the framework
of P~systems.  Using tree-based P~systems, Bernardini \emph{et al} \cite{Bernardini2008} 
provided a non-deterministic with time complexity $3h$ 
and a deterministic solution with time complexity $4n+2h$,
where $h$ is the height of the tree structure underlying the P~system and 
$n$ is the number of membranes of the P~system.  
The deterministic solution requires membrane \emph{polarization} techniques
and uses a \emph{depth-first-search}.

More recently, Alhazov \emph{et al} \cite{AlhazovMV2008} described 
an improved deterministic algorithm for tree-based P~systems, that runs in $3h+3$ steps.
This solution requires conditional rules (promoters and inhibitors)
and combines a \emph{breadth-first-search}, a \emph{broadcast} and a \emph{convergecast}, 
algorithmic techniques with a high potential for parallelism.

In this paper, we continue the study of FSSP in the framework of P~systems, 
by providing solutions for hyperdag P~systems \cite{Nicolescu2008} and for neural
P~systems \cite{Paun2002} with symmetric communication channels.  
We propose deterministic solutions to a variant of FSSP \cite{Szwerinski1982}, 
in which there is a single commander, at an arbitrary position.  
We further generalize this problem by synchronizing a subset of cells of the considered
hyperdag or neural P~system. 

These more complex structures pose additional challenges,
not considered in the previous FSSP papers on tree-based P~systems,
such as multiple network sources (no single root) and multiple paths between cells.
Additionally, by allowing an arbitrary position for the commander,
we cannot anymore take advantage of the sense of direction between adjacent cells;
practically, our structures need to be treated as undirected graphs.

Our first solution uses simple rules, but requires \emph{dynamical structures}. 
In this paper we propose a \emph{novel extension}, which supports 
the creation of dynamical structures, by allowing \emph{mobile channels}.
This solution works for hP~systems and symmetric nP~systems; 
it will also work for tree-based P~systems,
but only if we reconsider them as dag-based P~systems,
because the resulting structures will be dags, not trees.
This solution takes $e_c+5$ steps, 
where $e_c$ is the eccentricity of the commander cell of the underlying dag or digraph. 
The relative simplicity and the speed of this solution supports our hypothesis
that basing P~systems on dag, instead of tree, structures allows 
more natural expressions of some fundamental distributed algorithms 
\cite{Nicolescu2008,RaduWMC2009}.

Our second solution is more traditional and does not require dynamical structures,
but is substantially more complex, 
combining a \emph{breadth-first-search}, a \emph{broadcast} and a \emph{convergecast}.
This solution works for tree-based P~systems, hP~systems and symmetric nP~systems and takes $6e_c+7$ steps. 
When restricted to P~systems, our algorithm takes more steps than 
Alhazov \emph{et al} \cite{AlhazovMV2008}, if the commander is the root node, 
but comparable to this, when the commander is a central node of an unbalanced rooted tree.

Our two solutions do not require polarizations or conditional rules,
but require \emph{states}, as defined for hyperdag and neural P~systems

Section~\ref{sec:prelim} provides background definitions 
and introduces the families of P systems considered for synchronization.  
Next, in Section~\ref{sec:PSystemsFSSP}, we cite the communication models for
hyperdag P~systems and neural P~systems, and the
transition and rewrite rules available for solving the FSSP.  
Our two FSSP solutions are described in
Sections~\ref{sec:StructuralExtensions} and \ref{sec:SolutionViaRules},
where we also illustrate the evolution of our FSSP algorithms. 
Finally, we end with some concluding remarks.


\section{Preliminary}
\label{sec:prelim}

A (binary) \emph{relation} $R$ over two sets $X$ and $Y$  
is a subset of their Cartesian product, $R \subseteq X \times Y$.
For $A \subseteq X$ and $B \subseteq Y$, we set
$R(A) = \{ y \in Y \mid \exists x \in A, (x,y) \in R \}$,  
$R^{-1}(B) = \{ x \in X \mid \exists y \in B, (x,y) \in R \}$.

A \emph{digraph} (directed graph) $G$ is a pair $(X, A)$, 
where $X$ is a finite set of elements called \emph{nodes} (or \emph{vertices}), 
and $A$ is a binary relation $A \subseteq X \times X$, of elements called \emph{arcs}. 
A length $n-1$ \emph{path} is a sequence of $n$ distinct nodes $x_1, \dots, x_n$, 
such that $\{ (x_1,x_2), \dots, (x_{n-1},x_n) \} \subseteq A$. 
A \emph{cycle} is a path $x_1, \dots, x_n$, 
where $n \ge 1 \mbox{ and } (x_n,x_1) \in A$.
A digraph is \emph{symmetric} if its relation $A$ is symmetric, i.e.,
$(x_1, x_2) \in A \Leftrightarrow (x_2, x_1) \in A$.
By default, all digraphs considered in this paper, 
and all structures from digraphs (dag, rooted tree, see below) 
will be \emph{weakly connected}, i.e., 
each pair of nodes is connected via a chain of arcs, where the arc direction is not relevant.

A \emph{dag} (directed acyclic graph) is a digraph $(X, A)$ without cycles.
For $x \in X$, $A^{-1}(x)$ are $x$'s \emph{parents},
$A(x)$ are $x$'s \emph{children}, and $A(A^{-1}(x)) \backslash \{x\}$ are $x$'s \emph{siblings}.

A \emph{rooted tree} is a \emph{special case of dag}, where each node has exactly one parent,
except a distinguished node, called \emph{root}, which has none. 

Throughout this paper, we will use the term \emph{graph} to denote a \emph{symmetric digraph} 
and \emph{tree} to denote a \emph{rooted tree}.

For a given tree, dag or digraph, we define $e_c$, the \emph{eccentricity} of a node $c$,
as the maximum length of a shortest path between $c$ and 
any other reachable node in the corresponding structure. 

For a tree, the set of \emph{neighbors} of a node $x$, $Neighbor(x)$, is 
the union of $x$'s parent and $x$'s children.
For a dag $\delta$ and node $x$, we define 
$Neighbor(x) = \delta(x) \cup \delta^{-1}(x) \cup \delta(\delta^{-1}(x)) \backslash \{x\}$,
if we want to include the siblings, or, $Neighbor(x) = \delta(x) \cup \delta^{-1}(x)$, otherwise.
For a graph $G=(X,A)$, we set $Neighbor(x) = A(x) = \{y \mid (x,y) \in A\}$.  
Note that, as defined, $Neighbor$ is always a symmetric relation.

A special node $c$ of a structure will be designated as  the \emph{commander}.
For a given commander $c$, we define the \emph{level} of a node $x$, 
$level_c(x) \in \mathbb{N}$, as the length of a shortest path between the $c$ and $x$, 
over the $Neighbor$ relation. 

For a given tree, dag or digraph and commander $c$, for nodes $x$ and $y$,
if $y \in Neighbor(x)$ and $level_c(y)=level_c(x)+1$,
then $x$ is a \emph{predecessor} of $y$ and
$y$ is \emph{successor} of $x$.
Similarly, a node $z$ is  a \emph{peer} of a node $x$, 
if $z \in Neighbor(x)$ and $level_c(z)=level_c(x)$.
Note that, for a given node $x$, the set of peers and the set of successors are disjoint.
A node without a \emph{successor} will be referred to as a \emph{terminal}.
We define $maxlevel_c = max \{level_c(x) \mid x \in X\}$ and
we note $e_c = maxlevel_c$.
A \emph{level-preserving path} from $c$ to a node $y$ is a sequence $x_0, \ldots, x_k$, 
such that $x_0 = c, x_k = y, x_i \in Neighbor(x_{i-1}), level_c(x_i) = i, 1 \le i \le k$. 
We further define $count_c(y)$ as the number of distinct level-preserving paths from $c$ to $y$.
 
The level of a node and number of level-preserving paths to it can be determined by a standard breadth-first-search, 
as shown in Algorithm~\ref{alg:bfs}.
Intuitively, this algorithm defines a \emph{virtual dag} based on successor relation 
and, if the original structure is a tree, 
this algorithm will ``reset'' the root at another node in that tree. 

\begin{algorithm}{Determine levels and count level-preserving paths}
\label{alg:bfs}~

\begin{itemize}
\item INPUT: A tree, dag or digraph, with nodes $\{1,\ldots,n\}$ 
    and a commander $c \in \{1,\ldots,n\}$.

\item OUTPUT: The arrays $level_c[\,]$  and $count_c[\,]$ of shortest distances and number of level-preserving paths from $c$ to each node in the structure,
over the $Neighbor$ relation.
\end{itemize}
    
\hspace*{16pt}\begin{minipage}{5in}
\begin{tabbing}
\textbf{array} $level_c[1,\ldots,n] = [-1,\ldots,-1]$; $count_c[1,\ldots,n] = [0,\ldots, 0]$ \\
\textbf{queue} $Q = ()$\\
$Q \Leftarrow c$\\
$level_c[c] = 0$; $count_c[c]=1$\\
\textbf{while} \= $Q \neq ()$ \textbf{do}\\
    \> \= $x \Leftarrow Q$\\
    \> \textbf{for} \= \textbf{each} $y \in Neighbor(x)$ \textbf{do}\\
    \>    \> \textbf{if} \= $level_c[y] = -1$ \textbf{then}\\
    \>\>    \>	$Q \Leftarrow y$\\
    \>\>\>	    $level_c[y]=level_c[x]+1$\\
    \>    \> \textbf{if} \= $level_c[y] = level_c[x]+1$ \textbf{then}\\
    \>\>\>	    $count_c[y]=count_c[y]+count_c[x]$\\
\textbf{return} $level_c$
\end{tabbing}
\end{minipage}
\end{algorithm}

\begin{example}
Figures~\ref{fig:Tree-3}, \ref{fig:Dag-NoSiblings} and \ref{fig:Digraph-1} 
show $level_c$, \emph{predecessors}, \emph{successors}, \emph{peers} and $count_c$, 
for a tree, a dag and a digraph structure, respectively.
Small side-arrows indicate the arcs traversed while computing the levels, 
over the induced $Neighbor$ relation, as described in Algorithm~\ref{alg:bfs}.
\end{example}

\begin{figure}[h]
\begin{tabular}{lll}
\begin{minipage}{2.2in}
\centerline{\includegraphics[scale=0.9]{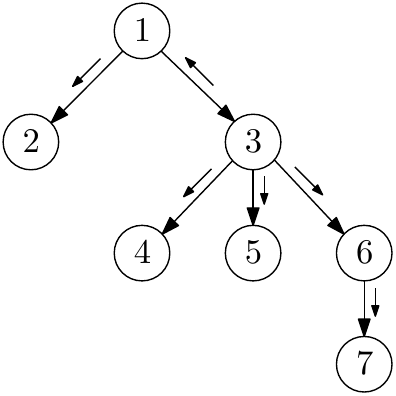}}
\end{minipage} &
\begin{minipage}{3in}
\begin{center}
\begin{tabular}{ | l | l | l | l | l | l | l | l | }
    \hline
    Node & $level_c$ & \emph{predecessors} & \emph{successors} & \emph{peers} & $count_c$ \\ \hline
    $1$ & $1$ & $3$ & $2$ & $-$ & $1$ \\ \hline
    $2$ & $2$ & $1$ & $-$ & $-$ & $1$ \\ \hline
    $3$ & $0$ & $-$ & $1,4,5,6$ & $-$ & $1$ \\ \hline
    $4$ & $1$ & $3$ & $-$ & $-$ & $1$ \\ \hline
    $5$ & $1$ & $3$ & $-$ & $-$ & $1$ \\ \hline
    $6$ & $1$ & $3$ & $7$ & $-$ & $1$ \\ \hline
    $7$ & $2$ & $6$ & $-$ & $-$ & $1$ \\ \hline
\end{tabular}
\end{center}
\end{minipage}
\end{tabular}
\caption{Left: a tree (taken from Bernardini \emph{et al} \cite{Bernardini2008}), 
with commander $c=3$, $e_3=2$; 
Right: table with node levels, predecessors, successors, peers and $count_c$'s.}
\label{fig:Tree-3}
\vspace*{-1.2cm}
\end{figure}

\begin{figure}[h]
\vspace*{1.0cm}
\begin{tabular}{ll}
\begin{minipage}{2.2in}
\centerline{\includegraphics[scale=0.85]{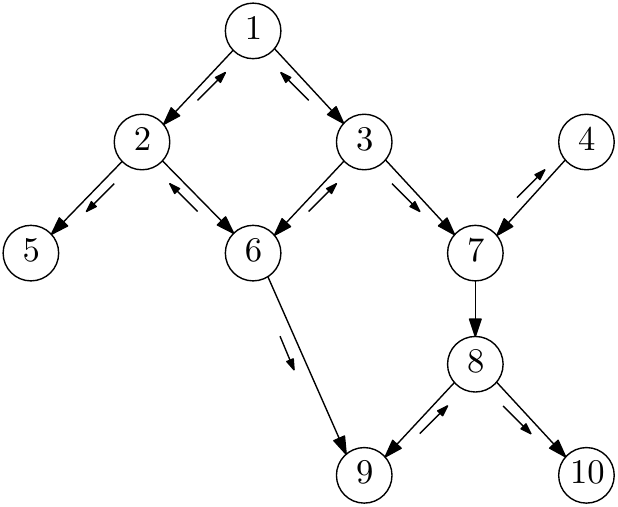}}
\end{minipage} &
\begin{minipage}{3in}
\begin{center}
\begin{tabular}{ | l | l | l | l | l | l | l | l | }
    \hline
    Node & $level_c$ & \emph{predecessors} & \emph{successors} & \emph{peers} & $count_c$ \\ \hline
    $1$ & $2$ & $2,3$ & $-$ & $-$ & $2$ \\ \hline
    $2$ & $1$ & $6$ & $1,5$ & $-$ & $1$ \\ \hline
    $3$ & $1$ & $6$ & $1,7$ & $-$ & $1$ \\ \hline
    $4$ & $3$ & $7$ & $-$ & $-$ & $1$ \\ \hline
    $5$ & $2$ & $2$ & $-$ & $-$ & $1$ \\ \hline
    $6$ & $0$ & $-$ & $2,3,9$ & $-$ & $1$ \\ \hline
    $7$ & $2$ & $3$ & $4$ & $8$ & $1$ \\ \hline
    $8$ & $2$ & $9$ & $10$ & $7$ & $1$ \\ \hline
    $9$ & $1$ & $6$ & $8$ & $-$ & $1$ \\ \hline
    $10$ & $3$ & $8$ & $-$ & $-$ & $1$ \\ \hline
\end{tabular}
\end{center}
\end{minipage}
\end{tabular}
\caption{Left: a dag with commander $c=6$, $e_6=3$ (siblings excluded);
Right: table with node levels, predecessors, successors, peers and $count_c$'s.}
\label{fig:Dag-NoSiblings}
\vspace*{-1in}
\end{figure}

\begin{figure}[h]
\vspace*{.9in}
\begin{tabular}{ll}
\begin{minipage}{2.2in}
\centerline{\includegraphics[scale=0.9]{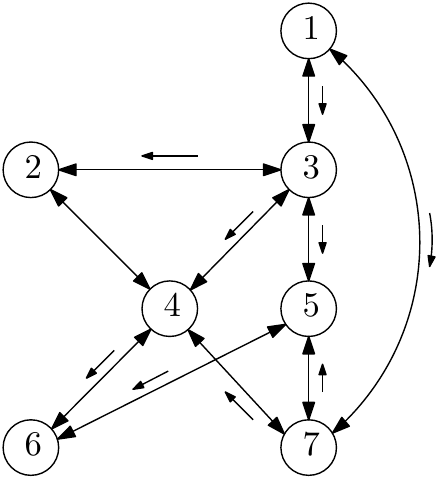}}
\end{minipage} &
\begin{minipage}{3in}
\begin{center}
\begin{tabular}{ | l | l | l | l | l | l | l | l | }
    \hline
    Node & $level_c$ & \emph{predecessors} & \emph{successors} & \emph{peers} & $count_c$ \\ \hline
    $1$ & $0$ & $-$ & $3,7$ & $-$ & $1$ \\ \hline
    $2$ & $2$ & $3$ & $-$ & $4$ & $1$ \\ \hline
    $3$ & $1$ & $1$ & $2,4,5$ & $-$ & $1$ \\ \hline
    $4$ & $2$ & $3,7$ & $6$ & $2$ & $2$ \\ \hline
    $5$ & $2$ & $3,7$ & $6$ & $-$ & $2$ \\ \hline
    $6$ & $3$ & $4,5$ & $-$ & $-$ & $4$ \\ \hline
    $7$ & $1$ & $1$ & $4,5$ & $-$ & $1$ \\ \hline
\end{tabular}
\end{center}
\end{minipage}
\end{tabular}
\caption{Left: a graph with commander $c=1$, $e_1=3$;
Right: table with node levels, predecessors, successors, peers and $count_c$'s.}
\label{fig:Digraph-1}
\vspace*{-1in}
\end{figure}


\clearpage

\section{P~Systems and the Firing Squad Synchronization Problem}
\label{sec:PSystemsFSSP}

In this section, we briefly recall several fundamental definitions for P~systems
and describe a P~systems version of the Firing Squad Synchronization Problem (FSSP).

For the definitions of tree-based P~systems, see P\u{a}un \cite{Paun2002}.
Here we reproduce the basic definitions of dag-based hyperdag P~systems,
from our previous work \cite{Nicolescu2008} and 
digraph-based neural P~systems, from P\u{a}un \cite{Paun2002}.
 
\begin{definition}[Hyperdag P~systems \cite{Nicolescu2008}]
	A \emph{hyperdag P~system} (of order $n$), in short an \emph{hP~system}, 
	is a system $\Pi_h = (O,\sigma_1,\dots,\sigma_n,\delta,I_{out})$, where:
	
	\begin{enumerate}
		\item $O$ is an ordered finite non-empty alphabet of \emph{objects};
	  \item $\sigma_1,\dots,\sigma_n$ are \emph{cells}, of the form 
					$\sigma_i = (Q_i,s_{i,0},w_{i,0},P_i)$, $1 \leq i \leq n$, where:  
				  \begin{itemize}
					  \item $Q_i$ is a finite set (of \emph{states}),
					  \item $s_{i,0} \in Q_i$ is the \emph{initial state},
					  \item $w_{i,0}\in O^*$ is the \emph{initial multiset} of objects,
					  \item	$P_i$ is a finite set of multiset \emph{rewrite rules} of the form: 
					  			$sx\rightarrow s'x'u_{\uparrow}v_{\downarrow}w_{\leftrightarrow}y_{go}z_{out}$, 
								  where $s,s'\in Q_i$, $x,x'\in O^{*}$, 
								  $u_{\uparrow}\in O_{\uparrow}^{*}$, 
								  $v_{\downarrow}\in O_{\downarrow}^{*}$, 
								  $w_{\leftrightarrow}\in O_{\leftrightarrow}^{*}$, 
								  $y_{go}\in O_{go}^{*}$ and $z_{out}\in O_{out}^{*}$, 
								  with the restriction that $z_{out}=\lambda$ for all $i\in \{1,\dots,n\} \backslash I_{out}$. 
				  \end{itemize}
		\item $\delta$ is a set of dag parent-child arcs on $\{1,\dots,n\}$,
					i.e., $\delta \subseteq \{1,\dots,n\} \times \{1,\dots,n\}$, 
					representing \emph{duplex} communication channels between cells;
		\item $I_{out} \subseteq \{1,\dots,n\}$ indicates the \emph{output cells}, 
					the only cells allowed to send objects to the ``environment''.
	\end{enumerate}
\end{definition}

\begin{definition}[Neural P~systems \cite{Paun2002}]
	A \emph{neural P~system} (of order $n\geq 1$), in short an \emph{nP~system}, 
	is a system $\Pi_n = (O,\sigma_1,\dots,\sigma_n,syn,i_{out})$, where:

	\begin{enumerate}
		\item $O$ is an ordered finite non-empty alphabet of \emph{objects};
		\item $\sigma_1,\dots,\sigma_n$ are \emph{cells}, of the form 
					$\sigma_i = (Q_i,s_{i,0},w_{i,0},P_i)$, $1\leq i\leq n$, where:
				  \begin{itemize}
					  \item $Q_i$ is a finite set (of \emph{states}),
					  \item $s_{i,0} \in Q_i$ is the \emph{initial state},
					  \item $w_{i,0}\in O^*$ is the \emph{initial multiset} of objects,
					  \item	$P_i$ is a finite set of multiset \emph{rewrite rules} 
								  of the form: $sx\rightarrow s'x'y_{go}z_{out}$, 
								  where $s,s'\in Q_i$, $x,x'\in O^{*}$, $y_{go}\in O_{go}^{*}$ and $z_{out}\in O_{out}^{*}$, 
								  with the restriction that $z_{out}=\lambda$ for all $i\in \{1,\dots,n\} \backslash \{i_{out}\}$. 
				  \end{itemize}	
		\item $syn$ is a set of \emph{digraph} arcs on $\{1,\dots,n\}$,
					i.e., $syn \subseteq \{1,\dots,n\} \times \{1,\dots,n\}$, 
					representing \emph{unidirectional} communication channels between cells, known as \emph{synapses};
		\item $i_{out} \in \{1,\dots,n\}$ indicates the \emph{output cell}, the only cell 
					allowed to send objects to the ``environment''.
	\end{enumerate}
\end{definition}

A \emph{symmetric} nP~system, (here) in short, a \emph{snP~system}, is an nP~system where the 
underlying digraph $syn$ is symmetric (i.e., a graph).
For further definitions describing the evolution of hP and nP~systems, 
such as \emph{configuration}, \emph{rewrite modes}, \emph{transfer modes}, \emph{transition steps}, 
\emph{halting} and \emph{results}, see our previous work \cite{Nicolescu2008}.
For all structures, we also utilize the \emph{weak policy} 
for applying \emph{priorities} to rules, as defined by P\u{a}un \cite{Paun2006}.

\begin{remark}
\label{rem:GoTag}
Most of the P~systems considered here (i.e., nP~systems, snP~systems, hP~systems with siblings and
hP~systems without siblings) define a tag $go$ that sends a multiset of objects along the 
previously defined $Neighbor$ relation.
Traditional tree-based P~systems do not directly provide this facility, however, 
it can be easily provided by the union of $out$ and $in!$ target indications, 
that represent sending ``to parent'' and ``to all children'', respectively.
That is, $(w,go) \equiv (w,out)(w,in!)$.
\end{remark}

\begin{definition}[FSSP for P~systems with states---informal definition]
\label{def:FSSP-P-B}
We are given a P, hP, snP or nP~system with $n$ cells, $\{ \sigma_1,\dots,\sigma_n \}$, 
where all cells have the \emph{same states set} and \emph{same rules set}.
Two states are distinguished: an \emph{initial} state $s_0$ and a \emph{firing} state $s_\phi$.
We select an arbitrary \emph{commander} cell $\sigma_c$ and an arbitrary subset of \emph{squad} cells, 
$F \subseteq \{\sigma_1,\dots,\sigma_n\}$ (possibly the whole set), that we wish to synchronize;
the commander itself may or may not be part of the firing squad.
At startup, all cells start in the initial state $s_0$; 
the commander and the squad cells may contain specific objects, but all other cells are empty. 
Initially, all cells, except the commander, are idle, 
and will remain idle until they receive a message.
The commander sends one or more orders, to one or more of its neighbors, 
to start and control the synchronization process.
Idle cells may become active upon receiving a first message.
Notifications may be relayed to all cells, as necessary.
Eventually, all cells in the squad set $F$ will enter the designated firing state $s_\phi$, 
\emph{simultaneously} and for the \emph{first time}. 
At that time, all the other cells have reached a different state, typically $s_0$ or $s_1$,
without ever passing through the firing state $s_\phi$.
Optionally, at that time, all cells should be empty.
\end{definition}

In this paper, we propose two new deterministic FSSP solutions,
that are described in the next two sections.
Our two solutions do not require \emph{polarities} or \emph{conditional} rules, 
but require \emph{priorities} and \emph{states}. 
Both hP~systems and snP~systems already have states, by definition.
However, it seems that traditional tree-based P~systems have not used states so far, or not much.


\section{FSSP---Dynamic Structures via Mobile Channels}
\label{sec:StructuralExtensions}

In this section, we further refine our solution given in an earlier paper \cite{RaduWMC2009}. 
A natural solution is possible when we are 
allowed to extend the cell structure of the given hP or snP~system. 
We achieve this by supporting mobile channels. 
The endpoints of our mobile channels appear in the rules like all other objects
and are subject to usual rewriting and transfer rules.
The end result is a channel that grows step-by-step, 
not unlike a nerve which extends in a growing or regrowing tissue. 

We first extend the original hP or nP~system by an external cell, 
which will be called the \emph{sergeant}. 
Next, this sergeant will send a self-replicating mobile endpoint, 
that will be repeatedly broadcasted, until all cells are reached.
A mobile endpoint will leave a fixed endpoint in a squad cell and 
disappear without trace from the other cells.
In the end, the structure will be extended with new channels which will
link the sergeant with all squad cells.
Finally, when there are no more structural changes, 
the sergeant, will send a firing command to all squad cells, 
prompting these cells to enter the firing state, all at the same time.

Our algorithm uses the following special objects, which can in principle be 
rewritten and transferred as all other traditional objects, but, 
at the same time, are also endpoints for dynamically created channels:
\begin{itemize}
\item $\alpha$ is here the fixed endpoint of all dynamically created channels (here we use only one $\alpha$)
\item $\theta$ is a mobile endpoint of a dynamically created channel (here this symbol is further processed
by rewriting and transfer rules)
\item $\omega$ is a new fixed endpoint of a dynamically created channel (here this symbol will remain fixed)
\end{itemize}
Briefly, the initial structure is dynamically extended by all arcs $(\sigma_a, \sigma_t)$,
where $\sigma_a$ is the cell that contains $\alpha$ and $\sigma_t$ is a cell that contains $\theta$ or $\omega$.

The following algorithm assumes that the first step has already been completed, i.e., 
the sergeant was already created by one of the existing cell creation or division rules, 
already available for P~systems.

\bigskip

\begin{algorithm}{FSSP---Dynamic structures via mobile channels}
\label{alg:dynamic}
\Precondition{ 
An hP or nP~system, with $n$ cells $\sigma_1, \dots, \sigma_n$, a commander cell $\sigma_c$
and a set of squad cells $F$ to be synchronized.
Additionally, we already have a sergeant cell, $\sigma_{n+1}$, linked to the commander by
one arc, $(\sigma_{n+1},\sigma_c)$, for hP~systems, 
or by two arcs, $(\sigma_{n+1},\sigma_c)$, $(\sigma_c,\sigma_{n+1})$, for snP~systems.

All cells start in the state $s_0$ and have the same rules. 
The state $s_\phi$ is the firing state.
Initially, the sergeant $\sigma_{n+1}$ is marked by one object $\alpha$,
and each squad cell is marked by one object $f$ 
(this can include the commander $\sigma_c$ or the sergeant $\sigma_{n+1}$, or both);
all other cells have no objects.
}
\Postcondition{
All cells in the set $F$ enter state $s_\phi$, 
simultaneously and for the first time, after $e_c+5$ steps,
where $e_c$ is the commander's eccentricity in the underlying graph.
All other cells enter state $s_1$, without ever passing through state $s_\phi$.
}

\Rules{
(rules are applied under the weak interpretation of priorities, 
in the rewrite mode $\alpha = min$ and transfer mode $\beta = repl$):
\begin{enumerate}
\begin{tabular}{ll}
\begin{minipage}{2.4in}
	\item $s_0 \alpha \rightarrow s_2 \alpha \theta_{go}$
	\item $s_0 f \theta \rightarrow s_4 \omega \theta_{go}$
	\item $s_0 \theta \rightarrow s_1 \theta_{go}$
  \item $s_1 \theta \rightarrow s_1$
	\item $s_2 \alpha \rightarrow s_3 \alpha$
\end{minipage}
\begin{minipage}{2.4in}
	\item $s_3 \theta \rightarrow s_3$
	\item $s_3 f \alpha \rightarrow s_4 f \alpha \phi \phi_{go}$
	\item $s_3 \alpha \rightarrow s_1 \alpha \phi_{go}$
	\item $s_4 f \phi \rightarrow s_\phi$
	\item $s_4 \theta \rightarrow s_4$
\end{minipage}
\end{tabular}
\end{enumerate}
}
\end{algorithm}

\begin{example}
\label{ex:sync}
Figure~\ref{fig:spider-sync} and Table~\ref{tab:trace} illustrate this algorithm for an hP~system 
based on the dag of Figure~\ref{fig:Tree-3}. 
Here, the commander cell is $\sigma_3$, the squad set is $F = \{\sigma_1, \dots, \sigma_5\}$
and this system's structure has already been extended by the sergeant cell $\sigma_8$ 
and the arc $(\sigma_8, \sigma_3)$. The mobile channels are represented by dotted arrows.

\begin{figure}[h]
\begin{tabular}{llll}
\begin{minipage}{1.5in}
\centerline{\includegraphics[scale=0.7]{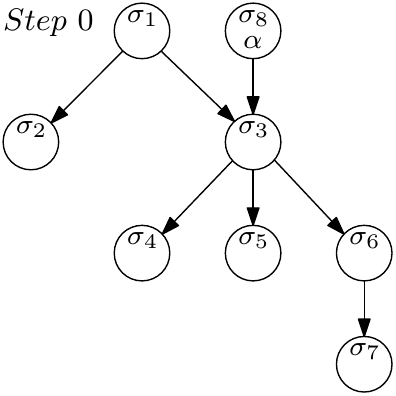}}
\end{minipage} 
\begin{minipage}{1.5in}
\centerline{\includegraphics[scale=0.7]{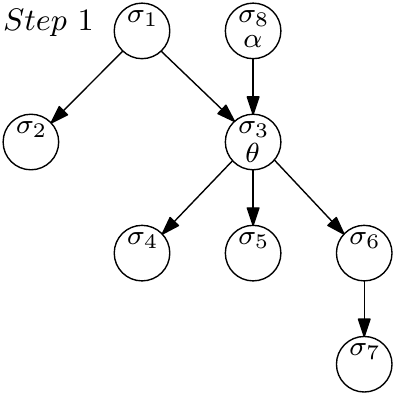}}
\end{minipage} 
\begin{minipage}{1.5in}
\centerline{\includegraphics[scale=0.7]{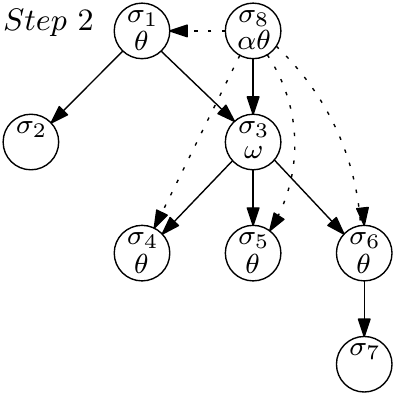}}
\end{minipage} 
\begin{minipage}{1.5in}
\centerline{\includegraphics[scale=0.7]{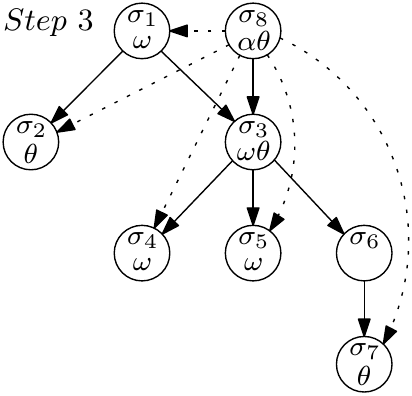}}
\end{minipage} 
\end{tabular}

\begin{tabular}{llll}
\begin{minipage}{1.5in}
\centerline{\includegraphics[scale=0.7]{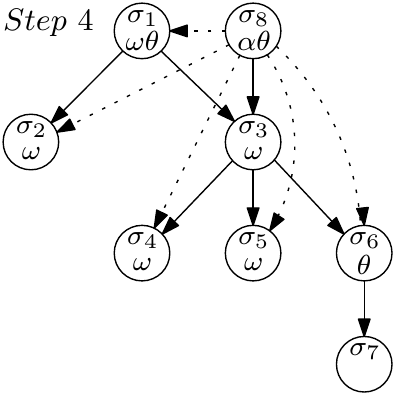}}
\end{minipage} 
\begin{minipage}{1.5in}
\centerline{\includegraphics[scale=0.7]{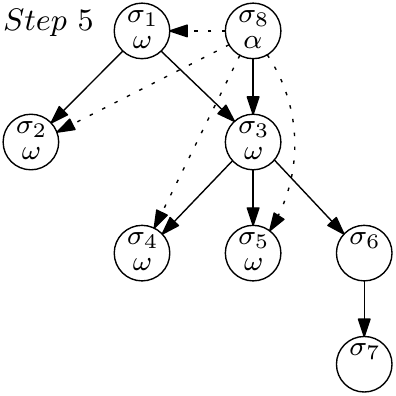}}
\end{minipage} 
\begin{minipage}{1.5in}
\centerline{\includegraphics[scale=0.7]{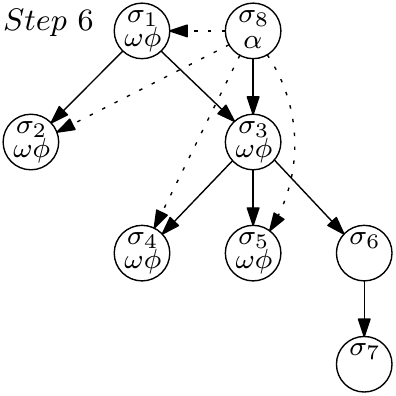}}
\end{minipage} 
\begin{minipage}{1.5in}
\centerline{\includegraphics[scale=0.7]{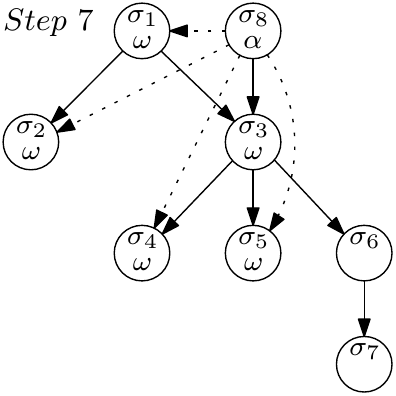}}
\end{minipage} 
\end{tabular}

\caption{Running Algorithm~\ref{alg:dynamic} on the hP~system of Example~\ref{ex:sync}.}
\label{fig:spider-sync}
\end{figure}

\begin{table}[h]
\begin{center}
\caption{Traces for the hP~system of Example~\ref{ex:sync}.}
\label{tab:trace}
\begin{tabular}{ | l | l | l | l | l | l | l | l | l | }
\hline                                                                                                                    
Step$\backslash$Cell & $\sigma_8$ & $\sigma_3$ & $\sigma_1$ & $\sigma_4$ & $\sigma_5$ & $\sigma_6$ & $\sigma_2$ & $\sigma_7$ \\ \hline
$0$ & $s_0 \alpha$ & $s_0 f$ & $s_0 f$ & $s_0 f$ & $s_0 f$ & $s_0$ & $s_0 f$ & $s_0$ \\ \hline
$1$ & $s_2 \alpha$ & $s_0 f \theta$ & $s_0 f$ & $s_0 f$ & $s_0 f$ & $s_0$ & $s_0 f$ & $s_0$ \\ \hline
$2$ & $s_3 \alpha \theta$ & $s_4 f \omega$ & $s_0 f \theta$ & $s_0 f \theta$ & $s_0 f \theta$ & $s_0 \theta$ & $s_0 f$ & $s_0$ \\ \hline
$3$ & $s_3 \alpha \theta^4$ & $s_4 f \omega \theta^4$ & $s_4 f \omega$ & $s_4 f \omega$ & $s_4 f \omega$ & $s_1$ & $s_0 f \theta$ & $s_0 \theta$ \\ \hline
$4$ & $s_3 \alpha \theta^2$ & $s_4 f \omega$ & $s_4 f \omega \theta$ & $s_4 f \omega$ & $s_4 f \omega$ & $s_1 \theta$ & $s_4 f \omega$ & $s_1$ \\ \hline
$5$ & $s_3 \alpha$ & $s_4 f \omega$ & $s_4 f \omega$ & $s_4 f \omega$ & $s_4 f \omega$ & $s_1$ & $s_4 f \omega$ & $s_1$ \\ \hline
$6$ & $s_1 \alpha$ & $s_4 f \omega \phi$ & $s_4 f \omega \phi$ & $s_4 f \omega \phi$ & $s_4 f \omega \phi$ & $s_1$ & $s_4 f \omega \phi$ & $s_1$ \\ \hline
$7$ & $s_1 \alpha$ & $s_\phi \omega$ & $s_\phi \omega$ & $s_\phi \omega$ & $s_\phi \omega$ & $s_1$ & $s_\phi \omega$ & $s_1$ \\ \hline
\end{tabular}
\end{center}
\end{table}

\end{example}


\bigskip

\section{FSSP---Static Structures and Rules}
\label{sec:SolutionViaRules}

Here we consider a second scenario, 
where we are allowed to modify the rules of the given hP or nP~system, 
but not its original structure.
A brief description of this solution follows.
The commander intends to send an order to all cells in the set $F$,
which will prompt them to synchronize by entering the designated firing state.
However, in general, the commander does not have direct communication channels with all the cells. 
In this case, the process of sending a command to the destination cell will cause delays (some steps),
as the command is relayed through intermediate cells. 
Hence, to ensure all firing squad cells enter the firing state simultaneously,
each firing squad cell determines the number of steps it needs to wait before 
entering the firing state.

As in our earlier paper \cite{RaduWMC2009}, 
cells have no built-in knowledge of the network topology.
Additionally, cells are anonymous, i.e., not identified by cell IDs, 
and not implicitly named by membrane polarization techniques.
The cells are initially empty, except the commander, which is initially marked by one $a$,
and the squad cells, which are initially marked by one $f$ each. 
All cells start with the same set of rules,
which are applied in the $max$ rewrite mode, using \emph{weak} priorities, and the $repl$ transfer mode.
In the proofs, all rules that are concurrently applied will be grouped together
within parentheses; e.g., $(x,y),z$ indicates two steps,
first rules $x$ and $y$, concurrently executed, followed by rule $z$.

Each cell \emph{independently} progresses through \emph{four phases}, 
called FSSP-I, FSSP-II, FSSP-III and FSSP-IV, which are detailed in  
Algorithms~\ref{alg:FSSP-I},~\ref{alg:FSSP-II},~\ref{alg:FSSP-III} and \ref{alg:FSSP-IV}, respectively. 
An overview of these four phases is as follows:

\begin{itemize}
\item Phase FSSP-I is a \emph{broadcast} from the commander, that follows the virtual dag defined by $level_c$.
      This phase starts in state $s_0$ and ends in state $s_2$.
			Also, the commander starts a counter, which, at the end of Phase FSSP-II, 
			will determine its eccentricity.
			
\item Phase FSSP-II is a subsequent \emph{convergecast} from terminal cells, 
			that follows the same virtual dag.
			This phase starts in state $s_2$ and ends when the commander enters state $s_6$.
			At the end of this phase, the commander's counter determines its eccentricity.

\item Phase FSSP-III is a second \emph{broadcast}, initiated from the commander, 
			that follows the same virtual dag. 
			This phase starts in state $s_6$ and ends in state $s_8$.
			The commander sends out its eccentricity, 
			which is successively decremented at each level.

\item Phase FSSP-IV is a \emph{timing} (countdown) for entering the firing state.
			This phase starts in state $s_8$ and continues with a countdown, until 
			squad cells simultaneously enter the firing state $s_9$, 
			and all other cells enter state $s_0$.
\end{itemize}

The statechart in Figure~\ref{fig:StateDiagram} illustrates the combined flow of these four phases.
The nodes represent the states of the hP or nP~system and the arcs are labelled 
with numbers of the rules that match the corresponding transitions.
The rest of this section describes these four phases, proving their correctness and time complexities.
A sample run of our algorithm will follow at the end of this section, in Example~\ref{ex:FSSP-trace}.

\begin{figure}[h]
\centerline{\includegraphics[scale=0.75]{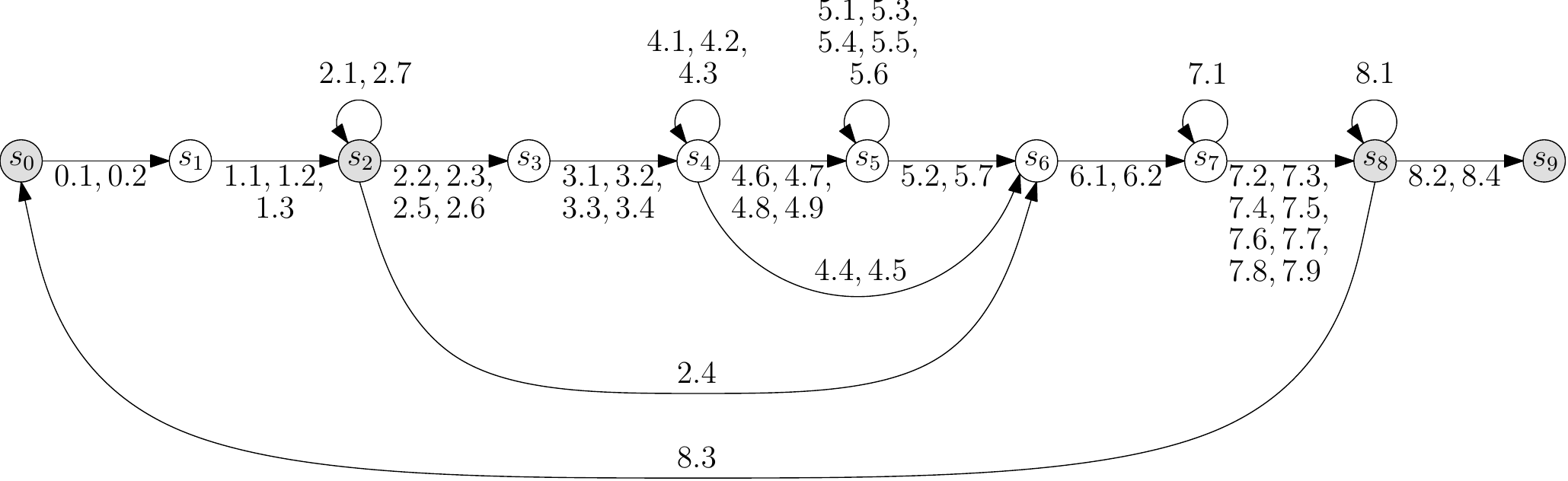}}
\caption{Statechart view of the combined FSSP algorithm phases.}
\label{fig:StateDiagram}
\end{figure}

\bigskip

\noindent \framebox{\textbf{FSSP: The initial configuration}}
\begin{itemize}
	\item $\Gamma = \{\sigma_1,\dots,\sigma_n\}$, $n>1$, is the set of all cells,
				$\sigma_c$ is the commander,
				and the firing squad is $F \subseteq \Gamma$;
	\item $O = \{a,b,c,d,e,f,g,h,k,l,p,q\}$;
	\item $Q_i = \{s_0, s_1, s_2, s_3, s_4, s_5, s_6, s_7, s_8, s_9\}$, for $i \in \{1,\dots,n\}$,
				which is ``allocated'' to four phases as follows:
  			FSSP-I contains rules for states $\{s_0, s_1\}$;
				FSSP-II contains rules for states $\{s_2, s_3, s_4, s_5, s_6\}$;
				FSSP-III contains rules for states $\{s_6, s_7\}$;
				FSSP-IV contains rules for states $\{s_8, s_9\}$;
	\item $s_\phi = s_9$ is the firing state;
	\item $s_{i,0} = s_0$, for $i \in \{1,\dots,n\}$;
	\item $w_{c,0} = \{a\}$, if $\sigma_c \notin F$, or $\{a,f\}$, otherwise;
				$w_{i,0} = \{f\}$ for all $\sigma_i \in F \setminus \sigma_c$;
				$w_{i,0} = \emptyset$, 
				for all $\sigma_i \in \Gamma \setminus ( F \cup \{ \sigma_c \} )$;			

	\item The following rules are applied under the weak interpretation of priorities, 
in the rewrite mode $\alpha = min$ and transfer mode $\beta = repl$:
	
		\begin{tabular}[t]{ll}
		\begin{minipage}[t]{2.0in}
			\begin{enumerate}
			\setcounter{enumi}{-1}
			\item For state $s_0$:
				\begin{enumerate}[1)]
					\item $s_0 a \rightarrow s_1 a e d_{go}$
					\item $s_0 d \rightarrow s_1 a d_{go}$
				\end{enumerate}
			\item For state $s_1$:
				\begin{enumerate}[1)]			
					\item $s_1 a e \rightarrow s_2 a ee k$
					\item $s_1 a \rightarrow s_2 a k$
					\item $s_1 d \rightarrow s_2 l$
				\end{enumerate}
			\item For state $s_2$:
				\begin{enumerate}[1)]
					\item $s_2 k \rightarrow s_2$
					\item $s_2 a e \rightarrow s_3 a ee$
					\item $s_2 d \rightarrow s_3 d$
					\item $s_2 a \rightarrow s_6 a c_{go}$
					\item $s_2 l \rightarrow s_3 l g_{go}$
					\item $s_2 g \rightarrow s_3$
					\item $s_2 a e \rightarrow s_2 a ee$
				\end{enumerate}
			\item For state $s_3$:
				\begin{enumerate}[1)]
					\item $s_3 a e\rightarrow s_4 a ee$
					\item $s_3 a \rightarrow s_4 a$
					\item $s_3 g \rightarrow s_4 p$
					\item $s_3 c \rightarrow s_4$
				\end{enumerate}
			\end{enumerate}
		\end{minipage}
		\begin{minipage}[t]{2.0in}
			\begin{enumerate}
			\setcounter{enumi}{3}
			\item For state $s_4$:
				\begin{enumerate}[1)]
					\item $s_4 c d \rightarrow s_4$
					\item $s_4 a d e \rightarrow s_4 a d ee$
					\item $s_4 d \rightarrow s_4 d$
					
					\item $s_4 a eeeee \rightarrow s_6 a eee$
					\item $s_4 eeeee \rightarrow s_6 e$
			
					\item $s_4 a \rightarrow s_5 a k$
					\item $s_4 l \rightarrow s_5 l h_{go}$
			
					\item $s_4 h \rightarrow s_5$
					\item $s_4 q \rightarrow s_5$
					
					\item $s_4 c \rightarrow s_6$ 
					\item $s_4 g \rightarrow s_6$ 
					\item $s_4 h \rightarrow s_6$ 
					\item $s_4 q \rightarrow s_6$ 
				\end{enumerate}
			\item For state $s_5$:
				\begin{enumerate}[1)]
					\item $s_5 k \rightarrow s_5$
					\item $s_5 a \rightarrow s_6 a c_{go}$
					\item $s_5 h p \rightarrow s_5 p$
					\item $s_5 p q \rightarrow s_5$
					\item $s_5 p \rightarrow s_5 k p$
					\item $s_5 l \rightarrow s_5 l h_{go}$
					\item $s_5 l \rightarrow s_6 q_{go}$
				\end{enumerate}
			\end{enumerate}
		\end{minipage}
		\begin{minipage}[t]{2.0in}
			\begin{enumerate}
			\setcounter{enumi}{5}
			\item For state $s_6$:
				\begin{enumerate}[1)]
					\item $s_6 a e \rightarrow s_7 a k$
					\item $s_6 e \rightarrow s_7 b e_{go}$
					
					\item $s_6 c \rightarrow s_6$ 
					\item $s_6 g \rightarrow s_6$ 
					\item $s_6 h \rightarrow s_6$ 
					\item $s_6 p \rightarrow s_6$ 
					\item $s_6 q \rightarrow s_6$ 
				\end{enumerate}
			\item For state $s_7$:
				\begin{enumerate}[1)]
					\item $s_7 k \rightarrow s_7$
					\item $s_7 a \rightarrow s_8 a$	
					\item $s_7 e \rightarrow s_8$ 
				\end{enumerate}
			\item For state $s_8$:
				\begin{enumerate}[1)]						
					\item $s_8 a b \rightarrow s_8 a$
					\item $s_8 a f \rightarrow s_9$
					\item $s_8 a \rightarrow s_0$
					\item $s_8 a \rightarrow s_9$
				\end{enumerate}
			\end{enumerate}
		\end{minipage}
		\end{tabular}
\end{itemize}


\begin{algorithm}{FSSP-I: First broadcast from the commander}
\label{alg:FSSP-I}
\Precondition{
The initial configuration as specified earlier.
}
\Postcondition{~
\begin{itemize}
	\item The end state is $s_2$.
	\item A cell $\sigma_i$ has 
	\begin{itemize}
		\item $count_c(i)$ copies of $a$ and $count_c(i)$ copies of $k$;
		\item $u$ copies of $l$, 
					where	$u$ is the total number of $a$'s in $\sigma_i$'s \emph{peers}; 
		\item $v$ copies of $d$, 
					where	$v$ is the total number of $a$'s in $\sigma_i$'s \emph{successors}; 
		\item two copies of $e$, if $\sigma_i = \sigma_c$;
		\item one copy of $f$, if $\sigma_i \in F$.
	\end{itemize}
\end{itemize}
}

\end{algorithm}

\begin{proof}
This phase of the algorithm is a \emph{broadcast} that follows the virtual dag 
created by the \emph{levels} determined by Algorithm~\ref{alg:bfs}. 

Consider a cell $\sigma_i$. By induction:
\begin{itemize}
\item At step $level_c(i)$, $\sigma_i$ (except the commander) 
receives a total of $count_c(i)$ copies of $d$ from its \emph{predecessors}.

\item At step $level_c(i)+1$, $\sigma_i$ broadcasts $count_c(i)$ copies of $d$ 
to each of its \emph{neighbors} and transits to state $s_1$.
At the same time, $\sigma_i$ accumulates one local copy of $a$ for each sent $d$,
for a total count of $count_c(i)$ of $a$'s.
Also, $\sigma_i$ receives $u$ copies of $d$, similarly sent by its \emph{peers}, 
where $u$ is equal to the total number of $a$'s similarly accumulated, at the same time step, 
by $\sigma_i$'s {peers}.

\item At step $level_c(i)+2$, $\sigma_i$ receives $v$ copies of $d$, sent back by its \emph{successors};
and transits to state $s_2$, 
where	$v$ is equal to the total number of $a$'s created, at the same time step, 
by $\sigma_i$'s {successors}; 
\end{itemize}

The commander, by initially having one $a$, creates two copies of $e$.
Finally, the rules associated with this phase do not change the number of $f$'s,
thus, each cell in the firing squad still ends with one $f$.
\end{proof}

\begin{corollary}[FSSP-I: Number of steps]
\label{cor:FSSP-I}
For each cell $\sigma_i$, the phase FSSP-I takes $level_c(i)+2$ steps.
\end{corollary}

\begin{proof}
As indicated in the proof of the Algorithm~\ref{alg:FSSP-I}, 
the total number of steps is $level_c(i)+2$.
\end{proof}

\medskip

\begin{algorithm}{FSSP-II: Convergecasts from terminal nodes}
\label{alg:FSSP-II}
\Precondition{
As described in the postcondition of Algorithm~\ref{alg:FSSP-I}.
}
\Postcondition{~
\begin{itemize}
	\item This phase ends when the commander enters state $s_6$.
	\item A cell $\sigma_i$ has
	\begin{itemize}
		\item $count_c(i)$ copies of $a$;
		\item $e_c+2$ copies of $e$, if $\sigma_i = \sigma_c$;
		\item one copy of $f$, if $\sigma_i \in F$.
	\end{itemize}
\end{itemize}
}


\end{algorithm}

\begin{proof}
Briefly, this phase of the algorithm is a \emph{convergecast} of $c$'s, starting from terminal cells,
and further relayed up, on the virtual dag, until the commander is reached.

For the purpose of this phase, the non-commander cells can be organized in the following three groups:
TC cells = terminal cells; 
NTC-NTP cells = non-terminal cells without non-terminal peers
(i.e., cells without peers or cells with terminal peers only);
NTC+NTP cells = non-terminal cells with non-terminal peers
(these cells may also have terminal peers).

During this phase, these cells will make transitions between the following three conceptual stages:
WCS = waiting for convergecasts from successors (state $s_4$);
RTC = ready to convergecast (state $s_5$);
HC = have convergecasted (state $s_6$).
Specifically, the following transitions will be made:
the TC cells will transit immediately from the WCS stage to the HC stage;
the NTC-NTP cells will linger in the WCS stage until they receive convergecasts from all their successors,
after which they will transit directly to the HC stage;
the NTC+NTP cells will linger in the WCS stage until they receive convergecasts from all their successors,
subsequently they will linger in the RTC stage until all their non-terminal peers reach the RTC stage as well,
after which they will transit to the HC stage.

During this process, cells will exchange $c$-notifications, which
are messages consisting of number of $c$'s and $h$-notifications,
which are messages consisting of number of $h$'s.
The actual numbers depend on network topology and take into account 
the multiple paths that appear in the virtual dag.

The $c$-notification broadcasted by cell $\sigma_i$ consists of $count_c(i)$ copies of $c$
and is only sent once when $\sigma_i$ transits into the HC stage.

The $h$-notification broadcasted by cell $\sigma_i$ consists of $u$ copies of $h$,
where $u$ is the number defined in the precondition.
This notification is sent repeatedly, while $\sigma_i$ remains in the RTC stage,
until $\sigma_i$ transits into the HC stage.
The $h$-notifications synchronize the non-terminal peers that cannot transit to the HC stage
until all of them have reached the RTC stage.
This avoids the potential confusion that could otherwise arise when a 
non-terminal cell receives an ``ambiguous'' $c$-notification,
i.e., a $c$-notification that could come both from a successor or from a non-terminal peer.

Without loss of generality, we illustrate our solution on the dag from Figure~\ref{fig:WLOG}.
This figure shows a typical sub-dag of the virtual dag created by Algorithm~\ref{alg:bfs},
where cells $\sigma_1$, $\sigma_2$, $\sigma_3$, $\sigma_4$ and $\sigma_5$ are at the same level,
and the horizontal lines indicate peer relations.

\begin{figure}[h]
\centerline{\includegraphics[scale=1.0]{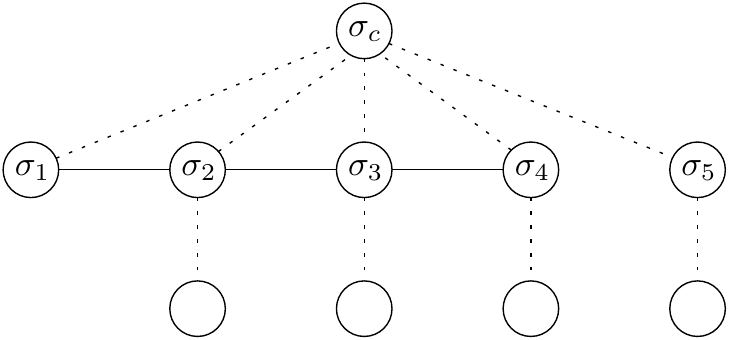}}
\caption{A typical sub-dag of the virtual dag.}
\label{fig:WLOG}
\end{figure}

Cell $\sigma_1$ is a TC cell and
will transit from stage WCS to stage HC, immediately after detecting that it has no successors
(there is no need for any synchronization with its peer $\sigma_2$).

Cell $\sigma_5$ is a NTC-NTP cell and will transit from stage WCS to stage HC, 
after receiving $c$-notifications from all its successors.

Cells $\sigma_2$, $\sigma_3$ and $\sigma_4$ are NTC+NTP cells.
Each of these cells will linger in stage WCS until it receives $c$-notifications from all its successors,
when it will enter stage RTC. Cells $\sigma_2$ and $\sigma_3$ are peers,
therefore none of them will be allowed to transit to stage HC,
until both of them have reached the RTC stage.
Similarly, cells $\sigma_3$ and $\sigma_4$ are peers,
therefore none of them will be allowed to transit to stage HC,
until both of them have reached the RTC stage.
Assume that cells $\sigma_2$, $\sigma_3$ and $\sigma_4$ will reach stage RTC in this order.
Then, cell $\sigma_2$ will wait in stage RTC until $\sigma_3$ also reaches the same stage.
When this eventuates, $\sigma_2$ will transit to stage HC, 
while $\sigma_3$ will still linger in stage RTC until $\sigma_4$ reaches the same stage.
When this eventuates, both $\sigma_3$ and $\sigma_4$ will transit at the same time to stage HC.

A TC cell $\sigma_i$ enters this phase $level_c(i)$ steps after the commander,
idles one step in state $s_2$, then starts its role in the convergecast, 
by broadcasting $count_c(i)$ copies of $c$ to its predecessors and peers
(it does not have successors) and transits to state $s_6$.
This cell further idles in state $s_6$ until it receives $e$'s from its predecessors.
The convergecast takes four steps at each level. 

The total run-time is dominated by $e_c$, the length of the longest level-preserving path from commander.
Therefore, the convergecast wave will complete at commander after $e_c + 4e_c - 2 = 5e_c - 2$ steps
after the commander starts this phase.
When the commander receives the convergecast from all its successors, 
it takes two steps to transit to state $s_6$.
Therefore, the commander enters state $s_6$, $5e_c$ steps after it starts this phase.
\end{proof}


\begin{corollary}[FSSP-II: Number of steps]
\label{cor:FSSP-II}
For each cell $\sigma_i$, the phase FSSP-II takes $5e_c - level_c(i)$ steps. 
\end{corollary}

\begin{proof}
As indicated in the proof of Algorithm~\ref{alg:FSSP-II},
this phase takes $5e_c$ steps.
\end{proof}


\newpage

\begin{algorithm}{FSSP-III: Second broadcast from the commander}
\label{alg:FSSP-III}
\Precondition{
As described in the postcondition of Algorithm~\ref{alg:FSSP-II}.
}
\Postcondition{~
\begin{itemize}
	\item The end state is $s_8$.
	\item A cell $\sigma_i$ has
	\begin{itemize}
		\item $count_c(i)$ copies of $a$;
		\item $(e_c + 1 - level_c(i))count_c(i)$ copies of $b$;
		\item one copy of $f$, if $\sigma_i \in F$.
	\end{itemize}
\end{itemize}
}

\end{algorithm}

\begin{proof}
In this phase, commander starts its second \emph{broadcast}, 
by sending $e_c+1$ copies of $e$'s to all its successors.
By induction on level, a cell $\sigma_i$ receives a total of 
$(e_c + 2 - level_c(i))count_c(i)$ copies of $e$'s from its predecessors,
reduces this count by $count_c(i)$ (i.e., the count of $a$'s),
forwards the remaining $(e_c + 1 - level_c(i))count_c(i)$ copies of $e$'s to all its successors
and creates for itself $(e_c + 1 - level_c(i))count_c(i)$ copies of $b$'s.
A more detailed description will be given in the final version.

All rules of this phase do not change the number of $a$'s or the number of $f$'s;
therefore, the corresponding postcondition holds. 

\end{proof}

\begin{corollary}[FSSP-III: Number of steps]
\label{cor:FSSP-III}
For each cell $\sigma_i$, the phase FSSP-III takes $level_c(i)+3$ steps.
\end{corollary}

\begin{proof}
As indicated in the proof of Algorithm~\ref{alg:FSSP-III},
this phase takes $level_c(i)+3$ steps.
\end{proof}


\begin{algorithm}{FSSP-IV: Timing for entering the firing state}
\label{alg:FSSP-IV}
\Precondition{
As described in the postcondition of Algorithm~\ref{alg:FSSP-III}.
}
\Postcondition{~
\begin{itemize}
\item The end state is $s_9$ for cells in the firing squad,
			or $s_0$, otherwise.
\item Each cell is empty.
\end{itemize}
}

\end{algorithm}

\begin{proof}
As long as $b$'s are present, a cell $\sigma_i$ performs a transition step that 
decreases the number of $b$'s by $count_c(i)$ (i.e., the number of $a$'s). 
This step will be repeated $(e_c + 1 - level_c(i))$ times,
as given by the initial ratio between the number of $b$'s, $(e_c + 1 - level_c(i))count_c(i)$, 
and the number of $a$'s, $count_c(i)$. This is the delay every cell needs to wait, 
before entering either the firing state $s_9$ or the initial state $s_0$.

Finally, in the last step, cell $\sigma_i$ enters $s_9$, if $\sigma_i$ has one $f$, 
or $s_0$, otherwise. At the same time, all existing objects are removed.
\end{proof}

\begin{corollary}[FSSP-IV: Number of steps]
\label{cor:FSSP-IV}
For each cell $\sigma_i$, the phase FSSP-IV takes $e_c + 2 - level_c(i)$ steps.
\end{corollary}

\begin{proof}
As indicated in the proof of Algorithm~\ref{alg:FSSP-IV},
this phase takes $(e_c + 1 - level_c(i))+1 = e_c + 2 - level_c(i)$.
\end{proof}

\begin{theorem}
\label{the:timeComplexity}
For each cell $\sigma_i$, the combined running time of the four phases 
Algorithm~\ref{alg:FSSP-I}, \ref{alg:FSSP-II}, \ref{alg:FSSP-III} and \ref{alg:FSSP-IV} is 
$6e_c + 7$, where $e_c$ is the eccentricity of the commander $\sigma_c$.
\end{theorem}

\begin{proof}
The result is obtained by summing the individual running times of the four phases, 
as given by Corollaries~\ref{cor:FSSP-I},~\ref{cor:FSSP-II}, \ref{cor:FSSP-III} and \ref{cor:FSSP-IV}:
$(level_c(i) + 2) + (5e_c - level_c(i)) + (level_c(i) + 3) + (e_c + 2 - level_c(i)) = 6e_c + 7.$
\end{proof}

\medskip

\begin{example}
\label{ex:FSSP-trace}
We present traces of the FSSP algorithm for the hP~system
given in Figure~\ref{fig:Dag-NoSiblings} in Table~\ref{tab:Dag-NoSiblings},
where the cells are ordered according to their \emph{levels}
and the starting states of phases FSSP-II, FSSP-III and FSSP-IV are highlighted.
\end{example}


\begin{table}[h]
\caption{The FSSP trace on the dag of Figure~\ref{fig:Dag-NoSiblings}, where
$c=6$, $e_6=3$, $F=\{ \sigma_1, \sigma_4, \sigma_5, \sigma_7, \sigma_9, \sigma_{10} \}$.}
\label{tab:Dag-NoSiblings}
\begin{center}
\renewcommand{\tabcolsep}{4.5pt}
\renewcommand{\arraystretch}{1.3}
\footnotesize
\noindent
\begin{tabular}{ | l | l | l | l | l | l | l | l | l | l | l | l |}
\hline
 & $\sigma_6$ & $\sigma_2$ & $\sigma_3$ & $\sigma_9$ & $\sigma_1$ & $\sigma_5$ & $\sigma_7$ & $\sigma_8$ & $\sigma_4$ & $\sigma_{10}$\\ \hline
0 & $s_0 af$ & $s_0 $ & $s_0 $ & $s_0 f$ & $s_0 $ & $s_0 f$ & $s_0 f$ & $s_0 $ & $s_0 f$ & $s_0 f$ \\ \hline
1 & $s_1 aef$ & $s_0 d$ & $s_0 d$ & $s_0 df$ & $s_0 $ & $s_0 f$ & $s_0 f$ & $s_0 $ & $s_0 f$ & $s_0 f$ \\ \hline
2 & \myway{$s_2 ad^{3}e^{2}fk$} & $s_1 a$ & $s_1 a$ & $s_1 af$ & $s_0 d^{2}$ & $s_0 df$ & $s_0 df$ & $s_0 d$ & $s_0 f$ & $s_0 f$ \\ \hline
3 & $s_2 ad^{3}e^{3}f$ & \myway{$s_2 ad^{3}k$} & \myway{$s_2 ad^{3}k$} & \myway{$s_2 adfk$} & $s_1 a^{2}$ & $s_1 af$ & $s_1 adf$ & $s_1 ad$ & $s_0 df$ & $s_0 df$ \\ \hline
4 & $s_3 ad^{3}e^{4}f$ & $s_2 ad^{3}$ & $s_2 ad^{3}$ & $s_2 adf$ & \myway{$s_2 a^{2}k^{2}$} & \myway{$s_2 afk$} & \myway{$s_2 adfkl$} & \myway{$s_2 adkl$} & $s_1 af$ & $s_1 af$ \\ \hline
5 & $s_4 ad^{3}e^{5}f$ & $s_3 ad^{3}$ & $s_3 ad^{3}$ & $s_3 adf$ & $s_2 a^{2}$ & $s_2 af$ & $s_2 adfl$ & $s_2 adl$ & \myway{$s_2 afk$} & \myway{$s_2 afk$} \\ \hline
6 & $s_4 ad^{3}e^{6}f$ & $s_4 ac^{3}d^{3}$ & $s_4 ac^{2}d^{3}g$ & $s_4 adfg$ & $s_6 a^{2}$ & $s_6 af$ & $s_3 adfgl$ & $s_3 adgl$ & $s_2 afg$ & $s_2 afg$ \\ \hline
7 & $s_4 ad^{3}e^{7}f$ & $s_4 a$ & $s_4 adg$ & $s_4 adfg$ & $s_6 a^{2}$ & $s_6 af$ & $s_4 acdflp$ & $s_4 acdlp$ & $s_6 afg$ & $s_6 afg$ \\ \hline
8 & $s_4 ad^{3}e^{8}f$ & $s_5 ak$ & $s_4 adg$ & $s_4 adfg$ & $s_6 a^{2}$ & $s_6 af$ & $s_4 aflp$ & $s_4 alp$ & $s_6 af$ & $s_6 af$ \\ \hline
9 & $s_4 ad^{3}e^{9}f$ & $s_5 a$ & $s_4 adgh$ & $s_4 adfgh$ & $s_6 a^{2}$ & $s_6 af$ & $s_5 afhklp$ & $s_5 ahklp$ & $s_6 afh$ & $s_6 afh$ \\ \hline
10 & $s_4 acd^{3}e^{10}f$ & $s_6 a$ & $s_4 adgh^{2}$ & $s_4 adfgh^{2}$ & $s_6 a^{2}c$ & $s_6 acf$ & $s_5 afhlp$ & $s_5 ahlp$ & $s_6 afh$ & $s_6 afh$ \\ \hline
11 & $s_4 ad^{2}e^{11}f$ & $s_6 a$ & $s_4 acdgh^{2}q$ & $s_4 acdfgh^{2}q$ & $s_6 a^{2}$ & $s_6 af$ & $s_6 acfhpq$ & $s_6 achpq$ & $s_6 acfq$ & $s_6 acfq$ \\ \hline
12 & $s_4 ad^{2}e^{12}f$ & $s_6 a$ & $s_4 agh^{2}q$ & $s_4 afgh^{2}q$ & $s_6 a^{2}$ & $s_6 af$ & $s_6 af$ & $s_6 a$ & $s_6 af$ & $s_6 af$ \\ \hline
13 & $s_4 ad^{2}e^{13}f$ & $s_6 a$ & $s_5 agk$ & $s_5 afgk$ & $s_6 a^{2}$ & $s_6 af$ & $s_6 af$ & $s_6 a$ & $s_6 af$ & $s_6 af$ \\ \hline
14 & $s_4 ad^{2}e^{14}f$ & $s_6 a$ & $s_5 ag$ & $s_5 afg$ & $s_6 a^{2}$ & $s_6 af$ & $s_6 af$ & $s_6 a$ & $s_6 af$ & $s_6 af$ \\ \hline
15 & $s_4 ac^{2}d^{2}e^{15}f$ & $s_6 a$ & $s_6 ag$ & $s_6 afg$ & $s_6 a^{2}c$ & $s_6 af$ & $s_6 acf$ & $s_6 ac$ & $s_6 af$ & $s_6 af$ \\ \hline
16 & $s_4 ae^{15}f$ & $s_6 a$ & $s_6 a$ & $s_6 af$ & $s_6 a^{2}$ & $s_6 af$ & $s_6 af$ & $s_6 a$ & $s_6 af$ & $s_6 af$ \\ \hline
17 & \myway{$s_6 ae^{5}f$} & \myway{$s_6 a$} & \myway{$s_6 a$} & \myway{$s_6 af$} & \myway{$s_6 a^{2}$} & \myway{$s_6 af$} & \myway{$s_6 af$} & \myway{$s_6 a$} & \myway{$s_6 af$} & \myway{$s_6 af$} \\ \hline
18 & $s_7 ab^{4}fk$ & $s_6 ae^{4}$ & $s_6 ae^{4}$ & $s_6 ae^{4}f$ & $s_6 a^{2}$ & $s_6 af$ & $s_6 af$ & $s_6 a$ & $s_6 af$ & $s_6 af$ \\ \hline
19 & $s_7 ab^{4}e^{9}f$ & $s_7 ab^{3}k$ & $s_7 ab^{3}k$ & $s_7 ab^{3}fk$ & $s_6 a^{2}e^{6}$ & $s_6 ae^{3}f$ & $s_6 ae^{3}f$ & $s_6 ae^{3}$ & $s_6 af$ & $s_6 af$ \\ \hline
20 & \myway{$s_8 ab^{4}f$} & $s_7 ab^{3}e^{6}$ & $s_7 ab^{3}e^{6}$ & $s_7 ab^{3}e^{2}f$ & $s_7 a^{2}b^{4}k^{2}$ & $s_7 ab^{2}fk$ & $s_7 ab^{2}e^{2}fk$ & $s_7 ab^{2}e^{2}k$ & $s_6 ae^{2}f$ & $s_6 ae^{2}f$ \\ \hline
21 & $s_8 ab^{3}f$ & \myway{$s_8 ab^{3}$} & \myway{$s_8 ab^{3}$} & \myway{$s_8 ab^{3}f$} & $s_7 a^{2}b^{4}$ & $s_7 ab^{2}f$ & $s_7 ab^{2}e^{3}f$ & $s_7 ab^{2}e^{3}$ & $s_7 abfk$ & $s_7 abfk$ \\ \hline
22 & $s_8 ab^{2}f$ & $s_8 ab^{2}$ & $s_8 ab^{2}$ & $s_8 ab^{2}f$ & \myway{$s_8 a^{2}b^{4}$} & \myway{$s_8 ab^{2}f$} & \myway{$s_8 ab^{2}f$} & \myway{$s_8 ab^{2}$} & $s_7 abf$ & $s_7 abf$ \\ \hline
23 & $s_8 abf$ & $s_8 ab$ & $s_8 ab$ & $s_8 abf$ & $s_8 a^{2}b^{2}$ & $s_8 abf$ & $s_8 abf$ & $s_8 ab$ & \myway{$s_8 abf$} & \myway{$s_8 abf$} \\ \hline
24 & $s_8 af$ & $s_8 a$ & $s_8 a$ & $s_8 af$ & $s_8 a^{2}$ & $s_8 af$ & $s_8 af$ & $s_8 a$ & $s_8 af$ & $s_8 af$ \\ \hline
25 & $s_9 $ & $s_0 $ & $s_0 $ & $s_9 $ & $s_0 $ & $s_9 $ & $s_9 $ & $s_0 $ & $s_9 $ & $s_9 $ \\ \hline
\end{tabular}
\end{center}
\end{table}


\clearpage

\section{Conclusion}
\label{sec:conclusion}
We have presented two new algorithms for the Firing Squad Synchronization Problem 
that operate on several families of P systems.
Out of the box, both algorithms work for hyperdag P~systems and symmetric neural P~systems.
The first algorithm is based on dynamic structures and highlights 
the merits of dags as underlying structures for P~systems.
To support the required dynamic structures, we propose 
an extended interpretation of P~systems which allows mobile channels, 
a solution which we believe is fully compatible with the existing P~systems rules.

The second algorithm, which is more complex, 
is applicable to P~systems with static membrane topologies and
is uniformly defined in terms of a structural $Neighbor$ relation.
These two algorithms do not require naming facilities, such as cell IDs or cell polarization
and handle a generalized version of the FSSP, 
where the commander can assume an arbitrary position and 
only a specified subset of the cells needs to be synchronized.

The work started in this paper leaves open several interesting problems.
Can we find simpler and more efficient solutions for hP~systems based on single-sourced dags?
Can we find simpler and more efficient solutions for hP or snP~systems using named cells (unique cell IDs)? 
Can we find a solution for arbitrary strongly-connected (non necessarily symmetric) nP~systems?  
What is relation between the mobile channels, which we have here proposed for P~systems,
and the support for mobile channels in the $\pi$-calculus? 


\section*{Acknowledgements}

The authors wish to thank John Morris and the three anonymous 
reviewers for detailed comments and feedback that helped us 
improve the paper.


\bibliographystyle{eptcs}

\end{document}